\documentclass[onecolumn,draftclsnofoot,12pt]{IEEEtran}
\usepackage{amsfonts}
\usepackage{mathrsfs}
\usepackage[noblocks]{authblk}
\usepackage[compress]{cite}
\usepackage{bm}
\usepackage{algorithm}
\usepackage{algorithmic}
\usepackage{amsmath,amssymb,epsfig,graphics,subfigure}
\usepackage{theorem}
\usepackage{array,color}
\usepackage[compress]{cite}
\usepackage{float}

\newtheorem{theorem}{Theorem}

\theoremheaderfont{\normalfont\bfseries}
\begin{document}

\title{The Role of Large-Scale Fading in Uplink Massive MIMO Systems}

\author{\IEEEauthorblockN {Ang Yang, Zunwen He, Chengwen Xing, Zesong Fei, and Jingming Kuang}
\thanks{The authors are with School of Information and Electronics,
Beijing Institute of technology, Beijing 100081, China
(email:taylorkingyang@163.com, hezunwen@bit.edu.cn,
chengwenxing@ieee.org, feizesong@bit.edu.cn, and
JMKuang@bit.edu.cn).}

}

\maketitle

\begin{abstract}
In this correspondence, we analyze the ergodic capacity of a large
uplink multi-user multiple-input multiple-output (MU-MIMO) system
over generalized-$K$ fading channels. In the considered scenario,
multiple users transmit their information to a base station equipped
with a very large number of antennas. Since the effect of fast
fading asymptotically disappears in massive MIMO systems,
large-scale fading becomes the most dominant factor for the ergodic
capacity of massive MIMO systems. Regarding this fact, in our work
we concentrate our attention on the effects of large-scale fading
for massive MIMO systems. Specifically, some interesting and novel
lower bounds of the ergodic capacity have been derived with both
perfect channel state information (CSI) and imperfect CSI.
Simulation results assess the accuracy of these analytical
expressions.
\end{abstract}


\IEEEpeerreviewmaketitle

\section{Introduction}

In order to satisfy the ever-increasing demands from explosive
wireless data services, a breakthrough in spectral efficiency is
expected for the next generation wireless systems just as 5G.
Regarding the great success of MIMO technologies
\cite{yindijing_1,shijin_AF}, massive MIMO or large MIMO
technologies have attracted a lot of attention recently as a
promising enabling technology to greatly boost the spectrum
efficiency \cite{shijin_massive,Larsson_uplink}. Furthermore, for 5G
network densification is also of great importance to improve network
capacity by increasing frequency reuse factor. Therefore, wireless
designers are faced with a challenging task, which is how to
determine the coverage range of massive MIMO base stations (BSs). In
other words, the channel fading characteristics of massive MIMO
should be carefully investigated.

Deploying very large numbers of antennas at the transmitter or
receiver, linear signal processing can substantially reduce the
jitters from various fast fading \cite{large_number}. However, as
the other part of the overall fading, the large-scale fading can not
be simply neglected in the massive MIMO \cite{Larsson_uplink}. It is
should be highlighted that to the best of the authors' knowledge, in
the existing works on massive MIMO systems only the effects of fast
fading have been investigated in detail \cite{Larsson_uplink}, and
the large-scale fading is simply assumed to be constant and known a
priori. In practical scenarios, the distributions of the large-scale
fading will largely vary in different scenarios, such as urban and
open areas. The work on the effect of the large-scale fading on
massive MIMO seems largely open up to date.

Motivated by this fact, in this correspondence we concentrate our
attention on  large-scale fading by considering the generalized-$K$
fading, which is a generic model that occurs when small-scale fading
is modeled via the Nakagami-$m$ distribution and large-scale fading
via the gamma distribution \cite{generalized_fading_1}. This model
has been demonstrated to effectively approximate most of the fading
and shadowing effects occurring in wireless channels, and also to be
analytically friendlier than the Nakagami-$m$/lognormal model
\cite{C_Zhong_Nakagami}. Moreover, we focus on the uplink of a
massive MU-MIMO system operating over generalized-$K$ fading, where
one BS equipped with $N$ antennas receives the information of $K$
single-antenna mobile users ($1 \ll K \ll N$). Different from
\cite{Larsson_uplink}, we focus on the following two fundamental
questions:

\noindent (1) What is the impact of large-scale
fading on the erogidc capacity of massive MIMO?

\noindent (2)
Can we provide analytical expressions for the ergodic capacity of
massive MIMO over generalized-$K$ fading?

To tackle these problems, new lower bounds for the ergodic capacity
of one user are derived for both perfect channel state information
(CSI) and imperfect CSI, as well as the average ergodic capacity of
all the users in the cell. Our analytical expressions are
substantiated via Monte Carlo simulations. It is shown by our result
that when the large-scale fading parameter $m<2$, the decrease of
$m$ will bring large reduction of the capacity of the system with
both perfect and imperfect CSI. This indicates that massive MIMO can
achieve high frequency reuse in the urban scenario.

\section{System model and Preliminaries}

In our work, we focus on  the uplink MU-MIMO system as shown in Fig.
\ref{system_model_central.eps}, which consists of one BS equipped
with $N$ antennas and $K$ single-antenna mobile users. The users
transmit their data to the BS in the same time-frequency resource
and the $N \times 1$ the received signal column vector at BS equals
to
\begin{align}
\label{received signals} {\bf{y}} = \sqrt {p_u} {\bf{G}}{\bf{x}} +
{\bf{w}},
\end{align}where the $K \times 1$ column vector $\sqrt {p_u}{\bf{x}}$ denotes the signal
transmitted by the $K$ users. Without loss of generality, the
average transmit power of each user is assumed to be $P$. In
addition, symbol ${\bf{G}}$ represents the $N \times K$ channel
matrix between the BS and the $K$ users and the $\{n,k\}^{\rm{th}}$
entry $g_{nk} \buildrel \Delta \over = [{\bf{G}}]_{nk}$ is the
channel coefficient between the $n$th antenna at the BS and the
$k$th user. The $K \times 1$ column vector ${\bf{w}}$ denotes the
additive zero-mean Gaussian white noise with unit variance.

In general, the channel matrix ${\bf{G}}$ is made of two kinds of
fading, i.e., fast fading and large-scale fading. Then each element
of ${\bf{G}}$ can be written as
\begin{align}\label{channel matrix}
{g_{nk}} = {h_{nk}}\sqrt {{\beta _{k}}} ,n = 1,2, \ldots ,N,k = 1,2,
\ldots ,K,
\end{align}where ${h_{nk}}$ is the fast fading coefficient between the $n$th antenna of the
BS and the $k$th user. Symbol $\sqrt {{\beta _{k}}}$ denotes the
large-scale fading from the $k$th user and the BS, which equals
\begin{align}\label{large-scale}
 {{\beta _{k}}}  =\mu_k /D_{k}^v, k = 1,2,
\ldots ,K,
\end{align}
where $D_{k}$ is the distance between the $k$th user to the BS, $v$
is the path-loss exponent with typical values ranging from 2 to 6. The
large-scale fading coefficient $\mu _{k}$ is modeled as independent
and identically distributed (i.i.d.) gamma random variable (RV), the
PDF of which can be expressed as
\begin{align}\label{large-scale 2}
p\left( {{\mu _k}} \right) = \frac{{{\mu _k}^{{m_k} - 1}}}{{\Gamma
\left( {{m_k}} \right){\Omega _k}^{{m_k}}}}\exp \left( { -
\frac{{{\mu _k}}}{{{\Omega _k}}}} \right),{\mu _k},{\Omega _k},{m_k}
> 0,
\end{align}where $\Gamma(\cdot)$ denotes the
gamma function. Moreover, denoting $E[\cdot]$ as the expectation of
a RV, $m_k$ and $\Omega _k= E[\mu_k]/{m_k}$ are the shape and scale
parameters of the gamma distribution, respectively. In practical
situations it is expected that  nonzero small and finite but large
values of $m$ for urban and open areas, respectively. The moderate
values of $m$ corresponds to suburban and rural areas.

\section{Ergotic capacity of the uplink MU-MIMO system}

In this section, we first derive the expressions of the ergotic capacity
of single user in the uplink MU-MIMO system to clarify the relationship between the
ergodic capacity and the large-scale fading clearly. Furthermore, the average ergotic capacity of all the
users in a cell is also analyzed to
describe the network performance. Since in
massive MIMO systems zero-forcing (ZF) detector tends to be optimal
\cite{Larsson_uplink}, it is adopted here as the beamformig strategy.

\subsection{Ergotic capacity of uplink MU-MIMO with perfect CSI}

\subsubsection{Ergotic capacity of the $k$th user}

With an ideal assumption that the BS has perfect CSI, the lower bound of the
capacity of the $k$th user in this case can be given by \cite[Eq.
(13)]{Larsson_uplink}
\begin{align}\label{capacity ref}
{C_{P,k}^L} = {\log _2}\left( {1 + {\beta _k}{p_u}(M-K)} \right).
\end{align} In the above expression (\ref{capacity ref}), the involved
large scale fading parameter ${\beta _k}$ is a variable, here we
take a further step to analyze the ergodic capacity over ${\beta
_k}$.

\begin{theorem}
Considering large-scale fading, the ergodic
capacity of the $k$th user with perfect CSI is
\begin{align}\label{capacity theorem}
E\left[ {{C_{P,k}^L}} \right] & = \frac{{\Omega _k^{}{{p_u}(M-K)}
{{m_k}  } }}{{\ln \left( 2 \right)D_k^v}} \ {}_3{F_1}\left(
{\begin{array}{*{20}{c}}
{{m_k} + 1,1,1}\\
2
\end{array}; - \frac{{\Omega _k^{}{p_u}(M-K)}}{{D_k^v}}} \right),
\end{align}where ${}_p{F_q}\left( {\begin{array}{*{20}{c}}
{{a_{1,}} \ldots ,{a_p}}\\
{{b_1}, \ldots ,{b_q}}
\end{array};z} \right)$ is the generalized hypergeometric
function.
\end{theorem}
\begin{proof}
See Appendix A.
\end{proof}

\subsubsection{Average ergotic capacity of all the users}
In the considered communication systems, the coverage area of a cell is model as a disc and the BS is located in
the center of the cell. The mobile users are located uniformly in the
cell with $R_0 < D_k <R$. The distribution of the users along the
radius of the cell is expressed as
\begin{align}\label{circle average 0}
f_d\left( x \right) = \frac{{2K}}{{{R^2-R_0^2}}}x.
\end{align}where $x$ is the distance between one user and the BS.

Exploiting Eqs. (\ref{capacity theorem}) and (\ref{circle average
0}), the lower bound of the average capacity of all the users in the
cell, noting as $\bar C_P^L$, can be averaged after deriving the
integral over the radius as
\begin{align}\label{circle average 1}
\bar C_P^L &= \frac{1}{K}\sum\limits_{k = 1}^K {E\left[ {C_{P,k}^L}
\right]} \notag \\
&\mathop  \to \limits^{a.s.} \frac{2}{{{R^2-R_0^2}}}\int_{{R_0}}^R
{x\frac{{\Omega _k^{}{p_u}(M - K) {{m_k} } }}{{\ln \left( 2
\right){x^v}}} \times {\;_3}{F_1}\left( {\begin{array}{*{20}{c}}
{{m_k} + 1,1,1}\\
2
\end{array}; - \frac{{\Omega _k^{}{p_u}(M - K)}}{{{x^v}}}} \right)}
dx.
\end{align}

Therefore the integral in the above expression (\ref{circle average
1}) becomes to the concern of the following work.

\begin{theorem}
Considering the distribution of large-scale fading, the average
ergodic capacity of all the users with perfect CSI is
\begin{align}\label{capacity average theorem}
\bar C_P^L &\mathop  \to \limits^{a.s.} \frac{{2\Omega _k^{}{p_u}(M
- K){m_k}}}{{\ln \left( 2 \right)\left( {{R^2} - R_0^2} \right)
\left( v - 2 \right)}}  \Bigg\{ \frac{1}{{{R_0}^{v -
2}}}{}_4{F_2}\left( {\begin{array}{*{20}{c}}
{\frac{{v - 2}}{v},{m_k} + 1,1,1}\\
{\frac{{v - 2}}{v} + 1,2}
\end{array};\frac{{ - \Omega _k^{}{p_u}(M - K)}}{{{R_0}^v}}} \right) \notag\\
& \quad - \frac{1}{{{R^{v - 2}}}}{}_4{F_2}\left(
{\begin{array}{*{20}{c}}
{\frac{{v - 2}}{v},{m_k} + 1,1,1}\\
{\frac{{v - 2}}{v} + 1,2}
\end{array};\frac{{ - \Omega _k^{}{p_u}(M - K)}}{{{R^v}}}}
\right)\Bigg\}.
\end{align}
\end{theorem}
\begin{proof}
See Appendix B.
\end{proof}

Perfect CSI is only an ideal assumption. In practice, CSI should be estimated via training sequences or pilots.  Due to limited length of training sequences and
time varying nature of wireless channels, channel estimation errors
are inevitable. In the following section, the performance with imperfect CSI is analyzed.

\subsection{Ergotic capacity of uplink MU-MIMO with imperfect CSI}

In this subsection we take a further step to derive several useful
expressions of the ergodic capacity of the massive MIMO systems with
imperfect CSI. In practice, the channel matrix $\bf{G}$ need to be
estimated  with the assistance of pilots \cite{gaofeifei_training}.
During the training phase in the coherence interval, mutually
orthogonal pilot signals of length $\tau$ are transmitted by
different users. The pilot sequences used by all the users are
denoted as a $\tau \times K$ matrix $\sqrt{\tau p_u} {\bf \Phi}$
($\tau \ge K$) with ${\bf \Phi}^{\rm H} {\bf \Phi}={\bf I}_K$.
Denoting ${\bf G}$ as the $M \times K$ channel matrix between the BS
and the $K$ users, the $M \times \tau$ received pilot matrix at the
BS is expressed as
\begin{align}\label{pilot 1}
{\bf Y}_p= \sqrt{\tau p_u} {\bf G} {\bf \Phi}^{\rm T} + {\bf N},
\end{align}where ${\bf N}$ is the $M \times \tau$ noise matrix with i.i.d. $CN (0, 1)$
elements. The MMSE estimate of $ {\bf G}$ is
\begin{align}\label{pilot 1}
{\bf{\hat G}} = \frac{1}{{\sqrt {\tau {p_u}} }}{{\bf{Y}}_p}{\bf{\Phi
\tilde D}} = \left( {{\bf{G}} + \frac{1}{{\sqrt {\tau {p_u}}
}}{\bf{W}}} \right){\bf{\tilde D}},
\end{align}where ${\bf{W}} \buildrel \Delta \over = {\bf N
\Phi^{*}}$ has i.i.d. $CN (0, 1)$ elements, as ${\bf \Phi}^{\rm H}
{\bf \Phi}={\bf I}_K$. Since ${\bf{D}} = diag\left\{ {{\beta _i}}
\right\},i \in \left[ {1,K} \right]$ denotes the favorable
propagation, we have matrix ${\bf{\tilde D}} \buildrel \Delta \over
= {\left( {\frac{1}{{\tau {p_u}}}{{\bf{D}}^{{\bf{ - 1}}}} +
{{\bf{I}}_K}} \right)^{ - 1}} = diag \left\{ {\frac{{\tau
{p_u}{\beta _i}}}{{\tau {p_u}{\beta _i} + 1}}} \right\}$, $i \in
\left[ {1,K} \right]$. Denoting the channel estimation error as
${{\bf{G}}_\Delta } = {\bf{\hat G}} - {\bf{G}}$, it can be seen that
the elements of the $i$th column of ${{\bf{G}}_\Delta } $ are RVs
with zero means and variances ${\frac{{\tau {p_u}{\beta _i}}}{{\tau
{p_u}{\beta _i} + 1}}}$, which will decrease the received signal to
noise ratio (SNR) at the BS.

\subsubsection{Ergotic capacity of the $k$th user}
With imperfect CSI, we begin with the derivation of the
ergodic capacity of the $k$th user by employing a lower bound given as
\cite[Eq. (42)]{Larsson_uplink}
\begin{align}\label{capacity ref im 1}
{C_{IP,k}} \ge {\log _2}\left( {1 + \frac{{\tau p_u^2\left( {M - K}
\right)\beta _k^2}}{{\left( {\tau p_u^{}\beta _k^{} + 1}
\right)\sum\limits_{i = 1}^K {\frac{{p_u^{}\beta _i^{}}}{{\tau
p_u^{}\beta _i^{} + 1}}}  + \tau p_u^{}\beta _k^{} + 1}}} \right).
\end{align}

However, the bound given in \cite[Eq. (42)]{Larsson_uplink} is
complex and not possible to calculate the integral over ${\beta
_k}$. Using the fact that $\frac{{p_u^{}\beta _i^{}}}{{\tau
p_u^{}\beta _i^{} + 1}} \le \frac{1}{\tau }$ and $1 + \frac{a}{b}
\ge \frac{{1 + a}}{b},a \ge 0,b \ge 1$, a novel lower bound is
proposed as
\begin{align}\label{capacity ref im}
{C_{IP,k}}
 \ge {\log _2}\left( {1 + \tau p_u^2\left( {M - K} \right)\beta _k^2} \right) - {\log _2}\left( {1 + \tau p_u^{}\beta _k^{}} \right) - {\log _2}\left( {1 + \frac{K}{\tau }} \right) \buildrel \Delta \over =
 C_{IP,k}^L.
\end{align}

Now we derive the integral over the large-scale fading ${\beta _k}$
in the above expression (\ref{capacity ref im}) and obtain the
following theorem.
\begin{theorem}
Taking the large-scale fading into account, the ergodic
capacity of the $k$th user with imperfect CSI is
\begin{align}\label{capacity theorem im}
&E\left[ {C_{IP,k}^L} \right]\notag\\& = \underbrace {\frac{{\tau
\Omega _k^2p_u^2\left( {M - K} \right){2^{{m_k} + 1}}\Gamma \left(
{\frac{{{m_k}}}{2} + 1} \right)\Gamma \left( {\frac{{{m_k}}}{2} +
\frac{3}{2}} \right)}}{{\sqrt \pi  \ln \left( 2 \right)\Gamma \left(
{{m_k}} \right)D_k^{2v}}}{\;_4}{F_1}\left( {\begin{array}{*{20}{c}}
{\frac{{{m_k}}}{2} + 1,\frac{{{m_k}}}{2} + \frac{3}{2},1,1}\\
2
\end{array}; - \frac{{4\tau \Omega _k^2p_u^2\left( {M - K} \right)}}{{D_k^{2v}}}} \right)}_{{{{\Xi _{\rm{1}}}\left( {{D_k},{m_k},{p_u},v} \right)}}}\notag\\
&\quad - \underbrace {\frac{{\tau \Omega _k^{}{p_u}{m_k}}}{{\ln
\left( 2 \right)D_k^v}}{\;_3}{F_1}\left( {\begin{array}{*{20}{c}}
{{m_k} + 1,1,1}\\
2
\end{array}; - \frac{{\tau \Omega _k^{}{p_u}}}{{D_k^v}}} \right)}_{{{{\Xi _{\rm{2}}}\left( {{D_k},{m_k},{p_u},v} \right)}}} - {\log _2}\left( {1 + \frac{K}{\tau }} \right).
\end{align}

\end{theorem}
\begin{proof}
See Appendix C.
\end{proof}

\subsubsection{Average ergotic capacity of all the users}
With the aid of Theorem 3, we are capable to derive the average
ergodic capacity of all the user with imperfect CSI, which can be
derived following the same logic of the perfect CSI case as
\begin{align}\label{average im}
\bar C_{IP}^L\mathop  \to \limits^{a.s.} \frac{2}{{{R^2} -
R_0^2}}\int_{R_0^2}^R x {\Xi _{\rm{1}}}\left( {x,{m_k},{p_u},v}
\right)dx - \frac{2}{{{R^2} - R_0^2}}\int_{R_0^2}^R x {\Xi _2}\left(
{x,{m_k},{p_u},v} \right)dx - {\log _2}\left( {1 + \frac{K}{\tau }}
\right).
\end{align}
Deriving the integral of Eq. (\ref{average im}) yields the following
theorem.
\begin{theorem}
Taking the effects of the large-scale fading into account, the
average ergodic capacity of all the users with imperfect CSI is
\begin{align}\label{capacity theorem im}
{{\bar C}_{IP}}&\mathop  \to \limits^{a.s.} \frac{{\tau \Omega _k^2p_u^2\left( {M - K} \right){2^{{m_k}}}\Gamma \left( {\frac{{{m_k}}}{2} + 1} \right)\Gamma \left( {\frac{{{m_k}}}{2} + \frac{3}{2}} \right)}}{{\left( {{R^2} - R_0^2} \right)\sqrt \pi  \ln \left( 2 \right)\Gamma \left( {{m_k}} \right)\left( v - 1\right)}}\notag \\
&\quad\times\Bigg\{ \frac{1}{{{R_0}^{2v - 2}}}{}_5{F_2}\left(
{\begin{array}{*{20}{c}}
{\frac{{v - 1}}{v},\frac{{{m_k}}}{2} + 1,\frac{{{m_k}}}{2} + \frac{3}{2},1,1}\\
{\frac{{v - 1}}{v} + 1,2}
\end{array};\frac{{ - 4\tau \Omega _k^2p_u^2\left( {M - K} \right)}}{{{R_0}^{2v}}}} \right)\notag\\
&\quad \quad \quad\quad- \frac{1}{{{R^{2v - 2}}}}{}_5{F_2}\left(
{\begin{array}{*{20}{c}}
{\frac{{v - 1}}{v},\frac{{{m_k}}}{2} + 1,\frac{{{m_k}}}{2} +  \frac{3}{2},1,1}\\
{\frac{{v - 1}}{v} + 1,2}
\end{array};\frac{{ - 4\tau \Omega _k^2p_u^2\left( {M - K} \right)}}{{{R^{2v}}}}} \right)\Bigg\}\notag \\
& \quad - \frac{{2\tau \Omega _k^{}{p_u}{m_k}}}{{\ln \left( 2
\right)\left( {{R^2} - R_0^2} \right)\left(v-2\right)}} \Bigg\{
\frac{1}{{{R_0}^{v - 2}}}{}_4{F_2}\left( {\begin{array}{*{20}{c}}
{\frac{{v - 2}}{v},{m_k} + 1,1,1}\\
{\frac{{v - 2}}{v} + 1,2}
\end{array};\frac{{ - \tau \Omega _k^{}{p_u}}}{{{R_0}^v}}} \right)\notag\\
& \quad \quad\quad\quad- \frac{1}{{{R^{v - 2}}}}{}_4{F_2}\left(
{\begin{array}{*{20}{c}}
{\frac{{v - 2}}{v},{m_k} + 1,1,1}\\
{\frac{{v - 2}}{v} + 1,2}
\end{array};\frac{{ - \tau \Omega _k^{}{p_u}}}{{{R^v}}}} \right)\Bigg\} - {\log _2}\left( {1 + \frac{K}{\tau }}
\right).
\end{align}

\end{theorem}
\begin{proof}
Due to space limitations, we will skip the proof, which can be
obtained following the similar logics of Appendices B and C.
\end{proof}

We note that all the expressions in the theorems of this section are
novel and can be used to predict the ergodic capacity of the large
uplink MU-MIMO systems over generalized-$K$ fading channels. It is
also worth noting that they are closed-form expressions, involving
only a finite number of basic operations such as summations of
logarithms, Gamma functions, and generalized hypergeometric
function, etc. This result can be used to determine the radius of the coverage of massive MIMO base stations.

\section{Numerical results}

In this section, simulation results are presented to
examine the impact of network parameters on the ergodic capacity of
the large uplink MU-MIMO systems. We assume that the users are
located uniformly in a cell of $1000$ meters while no user is closer
to the BS than $100$ meters, which means that $R=1000$ and
$R_0=100$. The distribution of the users is plotted in Fig.
\ref{User_location.eps}. The transmitted signals suffer from the
Nakagami-$m$ distributed fast fading (we set the fast fading
parameter $m=1$), and the gamma distributed large-scale fading (we
set $\Omega _k= 1/{m_k}$ and $m_1=m_2= \ldots =m_K =m$). We also
assume equal transmit power at each node and the value of $P$ is set
to be the received power at $500$ meters far from the transmitted
users. The length of the pilot is set as $\tau = K$. Furthermore,
the path loss exponent is set as $v = 3.6$ \cite{D_Tse}.

Fig. \ref{Capacity_different_N.eps} plots both the simulated and
numerical average capacity of all the users with different number of
the BS antennas $N$. The number of the users is set as $K=9$ and the
large-scale fading parameter is set as $m=3.3$. It can be seen that
the numerical curves accurately predict the simulated ones in the
large number of $N$ with both perfect and imperfect CSI. By
observing these curves, it is evident that the capacity increases
with $N$, which indicates that increasing the number of the BS
antennas brings an improved performance of system throughput. For
example, increasing $N$ from $100$ to $300$ brings a capacity
advantage of almost $20\%$ at $P=10dB$ with perfect CSI. It is also evident
that the lack of CSI will decrease the capacity of the system. For
example, the imperfect CSI brings about a capacity loss of $12\%$ at
$P=10dB$ with $N=250$.

Figs. \ref{Capacity_different_naka_m_antennas.eps} and
\ref{Capacity_different_naka_m_power.eps} plot the average capacity
of all the users with different large-scale fading parameters $m$.
We set $K=9$. It can be observed that increasing $m$ will result in
a better average capacity of all the users in both figures. For
example, increasing $m$ from $0.1$ to $1$ brings about a capacity
advantage of $93\%$ at $N=128, P=20dB$ with perfect CSI. Increasing
$m$ from $2$ to $5$ brings about a capacity advantage of $2\%$ at
$N=128, P=20dB$ with perfect CSI. Similar observations can be
obtained by the curves with imperfect CSI. These observations
demonstrate that with both perfect CSI and imperfect CSI, the
average ergodic capacity of all the users increases obversely as $m$
goes large when $m<2$, which is the urban scenario with heavy
shadowing; larger $m$ will brings little improvement of the capacity
when $m>2$, which is the suburban, rural, or flatland scenario.

\section{Conclusions}
In this correspondence, we analyzed the ergodic capacity of a large
uplink MU-MIMO system where multiple single-antenna users transmit
their information to a base station equipped with a very large number
of antennas. Over generalized-$K$ fading channels, several novel
lower bounds of the ergodic capacity of a single user were derived
for both perfect CSI and imperfect CSI, as well as the average
ergodic capacity of all the users in the cell. Simulation results
were used to validate our analytical expressions. We figured out
that large-scale fading will have large effect on the capacity of
the system with both perfect and imperfect CSI in the urban
scenario. Our result is of importance to design the coverage of massive MIMO base stations.

\appendices

\section{Proof of Theorem 1} Substituting Eqs. (\ref{large-scale}) and
(\ref{large-scale 2}) into Eq. (\ref{capacity ref}), the ergodic
capacity of the $k$th user can be expressed as
\begin{align}\label{perfect proof 1}
E\left[ {{C_{P,k}^L}} \right] = &\int_0^\infty  {{{\log }_2}\left(
{1 + \frac{{{p_u}(M - K)}}{{D_k^v}}{\mu _k}} \right)} \frac{{{\mu
_k}^{{m_k} - 1}}}{{\Gamma \left( {{m_k}} \right){\Omega
_k}^{{m_k}}}}\exp \left( { - \frac{{{\mu _k}}}{{{\Omega _k}}}}
\right)d{\mu _k}.
\end{align}Employing ${\log _2}\left( x \right) = \ln \left( x \right)/\ln
\left( 2 \right)$ and $\ln \left( {1 + x} \right) = \sum\limits_{i =
0}^\infty  {{{\left( { - 1} \right)}^{i }}{x^{i+1}}/({i+1})}$
\cite[1.511]{Table_of_Integrals}, we can rewrite the above
expression (\ref{perfect proof 1}) as
\begin{align}\label{perfect proof 2}
E\left[ {C_{P,k}^L} \right] =  - \frac{1}{{\ln \left( 2
\right)\Gamma \left( {{m_k}} \right){\Omega
_k}^{{m_k}}}}\sum\limits_{i = 0}^\infty  {\frac{1}{{i + 1}}} {\left(
{ - \frac{{{{{p_u}(M - K)}}}}{{D_k^v}}} \right)^{i + 1}}
\int_0^\infty  {{\mu _k}^{{m_k} + i}\exp \left( { - \frac{{{\mu
_k}}}{{{\Omega _k}}}} \right)d} {\mu _k}.
\end{align}

Using the formula given by \cite[3.326.2]{Table_of_Integrals}, the
above expression (\ref{perfect proof 2}) can be expressed as
\begin{align}\label{perfect proof 3}
E\left[ {C_{P,k}^L} \right] &=  - \frac{1}{{\ln \left( 2
\right)\Gamma \left( {{m_k}} \right)}}\sum\limits_{i = 0}^\infty
{\frac{1}{{i + 1}}} \Gamma \left( {{m_k} + i + 1} \right)  {\left( {
- \frac{{\Omega _k^{}{p_u}(M - K)}}{{D_k^v}}} \right)^{i + 1}}.
\end{align}Multiplying the same items on both the numerators and denominators in the
fractions does not change the equality and thus Eq. (\ref{perfect proof 3}) can be rewritten as
\begin{align}\label{perfect proof 4}
E\left[ {C_{P,k}^L} \right] &= \frac{{\Omega _k^{}{p_u}(M - K)\Gamma \left( {{m_k} + 1} \right)}}{{\ln \left( 2 \right)\Gamma \left( {{m_k}} \right)D_k^v}}\sum\limits_{i = 0}^\infty  {\frac{{\Gamma \left( {{m_k} + 1 + i} \right)}}{{\Gamma \left( {{m_k} + 1} \right)}}} \frac{{\Gamma \left( {i + 1} \right)\Gamma \left( {i + 1} \right)\Gamma \left( 2 \right)}}{{\Gamma \left( 1 \right)\Gamma \left( 1 \right)\Gamma \left( {i + 2} \right)}}\frac{{{{\left( { - \frac{{\Omega _k^{}{p_u}(M - K)}}{{D_k^v}}} \right)}^i}}}{{\Gamma \left( {i + 1} \right)}}\notag \\
 &= \frac{{\Omega _k^{}{p_u}(M - K)\Gamma \left( {{m_k} + 1} \right)}}{{\ln \left( 2 \right)\Gamma \left( {{m_k}} \right)D_k^v}}\sum\limits_{i = 0}^\infty  {} \frac{{{{\left( {{m_k} + 1}
\right)}_i}{{\left( 1 \right)}_i}{{\left( 1 \right)}_i}}}{{{{\left(
2 \right)}_i}}}\frac{{{{\left( { - \frac{{\Omega _k^{}{p_u}(M -
K)}}{{D_k^v}}} \right)}^i}}}{{i!}},
\end{align}where ${\left( m
\right)_i} = \Gamma \left( {m + i} \right)/\Gamma \left( m \right)$
is named as the Pochmann symbol. Based on
\cite[16.2.1]{NIST_handbook_of_mathematical_functions}, Eq.
(\ref{perfect proof 4}) can be reformulated as Eq. (\ref{capacity
theorem}) and then the theorem can be achieved.

\section{Proof of Theorem 2}
Substituting Eq. (\ref{perfect proof 3}) into Eq. (\ref{circle
average 1}) and using the formula $\int {{x^a}} dx = {x^{a +
1}}/\left( {a + 1} \right)$, we can reformulate (\ref{circle average
1}) as
\begin{align}\label{average perfect proof 1}
\bar C_P^L &\mathop  \to \limits^{a.s.} \frac{{{\rm{2}}\Omega _k^{}{p_u}(M - K)\Gamma \left( {{m_k} + 1} \right)\Gamma \left( {\frac{{v - 2}}{v}} \right)}}{{v\ln \left( 2 \right)({R^2} - R_0^2)\Gamma \left( {{m_k}} \right)\Gamma \left( {\frac{{v - 2}}{v} + 1} \right)}}\sum\limits_{i = 0}^\infty  {\frac{{\Gamma \left( {i + 1} \right)\Gamma \left( {i + 1} \right)\Gamma \left( 2 \right)\Gamma \left( {i + \frac{{v - 2}}{v}} \right)}}{{\Gamma \left( 1 \right)\Gamma \left( 1 \right)\Gamma \left( {i + 2} \right)\Gamma \left( {\frac{{v - 2}}{v}} \right)}}}\notag \\
& \quad\times \frac{{\Gamma \left( {\frac{{v - 2}}{v} + 1}
\right)\Gamma \left( {{m_k} + i + 1} \right)}}{{\Gamma \left( {i +
\frac{{v - 2}}{v} + 1} \right)\Gamma \left( {{m_k} + 1}
\right)}}\left( {\frac{1}{{{R_0}^{v - 2}}}\frac{{{{\left( {\frac{{ -
\Omega _k^{}{p_u}(M - K)}}{{{R_0}^v}}} \right)}^i}}}{{i!}} -
\frac{{\rm{1}}}{{{R^{v - 2}}}}\frac{{{{\left( {\frac{{ - \Omega
_k^{}{p_u}(M - K)}}{{{R^v}}}} \right)}^i}}}{{i!}}} \right).
\end{align}Based on the definition of Pochmann symbol, the above expression
(\ref{average perfect proof 1}) can be expressed as
\begin{align}\label{average perfect proof 2}
\bar C_P^L &\mathop  \to \limits^{a.s.} \frac{{{\rm{2}}\Omega _k^{}{p_u}(M - K)\Gamma \left( {{m_k} + 1} \right)\Gamma \left( {\frac{{v - 2}}{v}} \right)}}{{v\ln \left( 2 \right)({R^2} - R_0^2)\Gamma \left( {{m_k}} \right)\Gamma \left( {\frac{{v - 2}}{v} + 1} \right)}}\notag\\
 &\quad\times \Bigg\{ \frac{1}{{{R_0}^{v - 2}}}\sum\limits_{i = 0}^\infty  {\frac{{{{\left( 1 \right)}_i}{{\left( 1 \right)}_i}{{\left( {\frac{{v - 2}}{v}} \right)}_i}{{\left( {{m_k} + 1} \right)}_i}}}{{{{\left( 2 \right)}_i}{{\left( {\frac{{v - 2}}{v} + 1} \right)}_i}}}} \frac{{{{\left( {\frac{{ - \Omega _k^{}{p_u}(M - K)}}{{{R_0}^v}}} \right)}^i}}}{{i!}} \notag\\
&\quad \quad\quad\quad- \frac{{\rm{1}}}{{{R^{v - 2}}}}\sum\limits_{i
= 0}^\infty {\frac{{{{\left( 1 \right)}_i}{{\left( 1
\right)}_i}{{\left( {\frac{{v - 2}}{v}} \right)}_i}{{\left( {{m_k} +
1} \right)}_i}}}{{{{\left( 2 \right)}_i}{{\left( {\frac{{v - 2}}{v}
+ 1} \right)}_i}}}} \frac{{{{\left( {\frac{{ - \Omega _k^{}{p_u}(M -
K)}}{{{R^v}}}} \right)}^i}}}{{i!}} \Bigg\},
\end{align}based on which and together with
\cite[16.2.1]{NIST_handbook_of_mathematical_functions}, Eq.
(\ref{capacity average theorem}) is obtained and then proof is
ended.

\section{Proof of Theorem 3}

Following a similar logic as the derivation of Eq. (\ref{perfect
proof 3}), the expectation of first item of Eq.
(\ref{capacity ref im}) over $\beta_k$ can be derived to be
\begin{align}\label{imperfect proof 1}
&{{\Xi _{\rm{1}}}\left(
{{D_k},{m_k},{p_u},v} \right)}\notag\\ &\buildrel \Delta \over = E\left[ {{{\log }_2}\left( {1 + \tau p_u^2\left( {M - K} \right)\beta _k^2} \right)} \right]\notag \\
& =  - \frac{1}{{\ln \left( 2 \right)\Gamma \left( {{m_k}} \right){\Omega _k}^{{m_k}}}}\sum\limits_{i = 0}^\infty  {\frac{1}{{i + 1}}} {\left( { - \frac{{\tau p_u^2(M - K)}}{{D_k^{2v}}}} \right)^{i + 1}} \int_0^\infty  {{\mu _k}^{{m_k} + 2i+1}\exp \left( { - \frac{{{\mu _k}}}{{{\Omega _k}}}} \right)d} {\mu _k}\notag \\
& = \frac{{\tau \Omega _k^2p_u^2(M - K)}}{{\ln \left( 2
\right)\Gamma \left( {{m_k}} \right)D_k^{2v}}}\sum\limits_{i =
0}^\infty  {\frac{1}{{i + 1}}}\Gamma \left( {{m_k} + 2i + 2}
\right){\left( { - \frac{{\tau \Omega _k^2p_u^2(M -
K)}}{{D_k^{2v}}}} \right)^i},
\end{align}where the last step is obtained by using
\cite[3.326.2]{Table_of_Integrals}. We apply the expression given by \cite[8.335.1]{Table_of_Integrals}
to derive the integral in (\ref{imperfect proof 1}), which can be
written as
\begin{align}\label{imperfect proof 2}
\Gamma \left( {2x} \right) = \frac{{{2^{2x - 1}}}}{{\sqrt \pi
}}\Gamma \left( x \right)\Gamma \left( {x + \frac{1}{2}} \right).
\end{align}Substituting (\ref{imperfect proof 2}) into (\ref{imperfect proof
1}), $I_1$ can be reformulated as
\begin{align}\label{imperfect proof 3}
{{\Xi _{\rm{1}}}\left(
{{D_k},{m_k},{p_u},v} \right)}& = \frac{{\tau \Omega _k^{2}p_u^2\left( {M - K} \right){2^{{m_k}+1}}}}{{\sqrt \pi  \ln \left( 2 \right)\Gamma \left( {{m_k}} \right)D_k^{2v}}}\sum\limits_{i = 0}^\infty  {\frac{1}{{i + 1}}} \Gamma \left( {i + \frac{{{m_k}}}{2} + 1} \right)\notag \\
& \quad\times \Gamma \left( {i + \frac{{{m_k}}}{2} +  \frac{3}{2}}
\right){\left( { - \frac{{4\tau \Omega _k^{2}p_u^2\left( {M - K}
\right)}}{{D_k^{2v}}}} \right)^i}.
\end{align} Multiplying the same items on both the numerators and denominators in the
fractions, the equality does not change and Eq. (\ref{imperfect proof 3}) can be further rewritten as
\begin{align}\label{imperfect proof 4}
{{\Xi _{\rm{1}}}\left(
{{D_k},{m_k},{p_u},v} \right)}& = \frac{{\tau \Omega _k^2p_u^2\left( {M - K} \right){2^{{m_k} + 1}}\Gamma \left( {\frac{{{m_k}}}{2} + 1} \right)\Gamma \left( {\frac{{{m_k}}}{2} + \frac{3}{2}} \right)}}{{\sqrt \pi  \ln \left( 2 \right)\Gamma \left( {{m_k}} \right)D_k^{2v}}}\notag \\
& \quad\times \sum\limits_{i = 0}^\infty  {} \frac{{\Gamma \left( {i + \frac{{{m_k}}}{2} + 1} \right)\Gamma \left( {i + \frac{{{m_k}}}{2} +  \frac{3}{2}} \right)\Gamma \left( {i + 1} \right)\Gamma \left( {i + 1} \right)}}{{\Gamma \left( {\frac{{{m_k}}}{2} + 1} \right)\Gamma \left( {\frac{{{m_k}}}{2} +  \frac{3}{2}} \right)\Gamma \left( 1 \right)\Gamma \left( 1 \right)}}\notag \\
& \quad\times \frac{{\Gamma \left( 2 \right)}}{{\Gamma \left( {i + 2} \right)}}\frac{{{{\left( { - \frac{{4\tau \Omega _k^{2}p_u^2\left( {M - K} \right)}}{{D_k^{2v}}}} \right)}^i}}}{{\Gamma \left( {i + 1} \right)}}\notag \\
& = \frac{{\tau \Omega _k^2p_u^2\left( {M - K} \right){2^{{m_k} + 1}}\Gamma \left( {\frac{{{m_k}}}{2} + 1} \right)\Gamma \left( {\frac{{{m_k}}}{2} + \frac{3}{2}} \right)}}{{\sqrt \pi  \ln \left( 2 \right)\Gamma \left( {{m_k}} \right)D_k^{2v}}}\notag \\
& \quad\times \sum\limits_{i = 0}^\infty  {} \frac{{{{\left(
{\frac{{{m_k}}}{2} + 1} \right)}_i}{{\left( {\frac{{{m_k}}}{2} +
\frac{3}{2}} \right)}_i}{{\left( 1 \right)}_i}{{\left( 1
\right)}_i}}}{{{{\left( 2 \right)}_i}}}\frac{{{{\left( { -
\frac{{4\tau \Omega _k^{2}p_u^2\left( {M - K} \right)}}{{D_k^{2v}}}}
\right)}^i}}}{{i!}}.
\end{align}Applying \cite[16.2.1]{NIST_handbook_of_mathematical_functions} to
Eq. (\ref{imperfect proof 4}), the expectation of first term over
large-scale fading can be expressed as ${{\Xi _{\rm{1}}}\left(
{{D_k},{m_k},{p_u},v} \right)}$ of Eq. (\ref{capacity theorem im}).
Since ${\log _2}\left( {1 + \tau p_u^{}\beta _k^{}} \right)$ is
similar to Eq. (\ref{capacity ref}), the expectation of second term
of Eq. (\ref{capacity ref im}) over $\beta_k$ can be also expressed
as ${{\Xi _{\rm{2}}}\left( {{D_k},{m_k},{p_u},v} \right)}$ of Eq.
(\ref{capacity theorem im}). Thus the proof is ended.

\newpage

\begin{figure}[!t]
\centering
\includegraphics[width=2.8in]{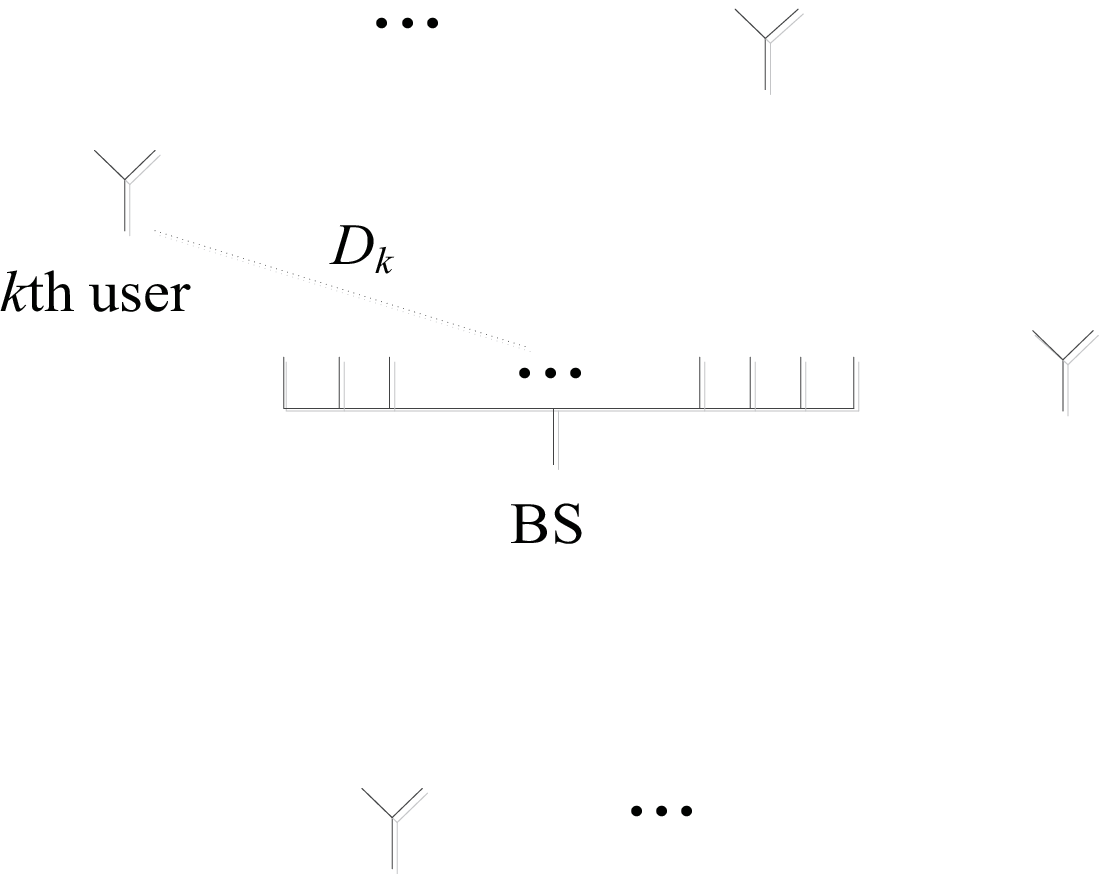}
\caption{The system model of uplink MU-MIMO system.}
\label{system_model_central.eps}
\end{figure}

\begin{figure}[!t]
\centering
\includegraphics[width=3.6in]{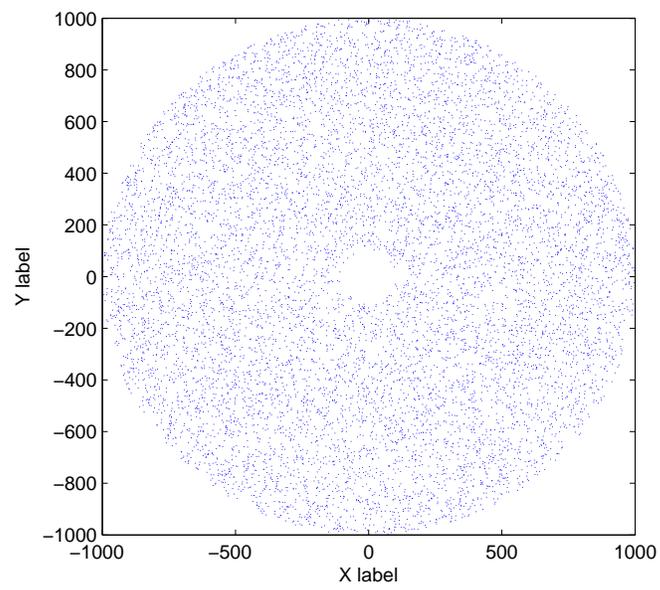}
\caption{The distribution of the users in the cell.}
\label{User_location.eps}
\end{figure}

\begin{figure}[!t]
\centering
\includegraphics[width=3.6in]{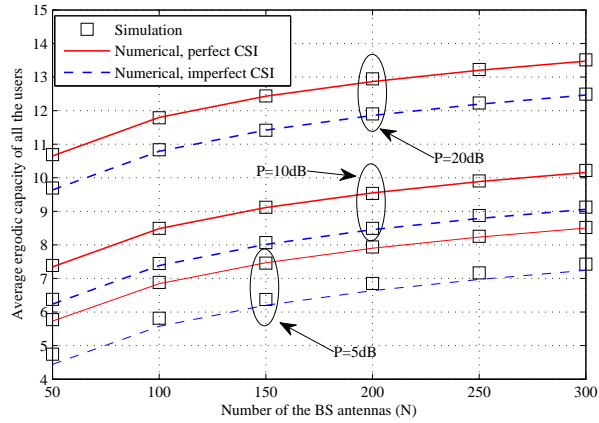}
\caption{The average ergodic capacity of all the users with
different number of BS antennas $N$ in the case of $K=9$, $m=3.3$,
and $P=5dB, 10dB, 20dB$.} \label{Capacity_different_N.eps}
\end{figure}

\begin{figure}[!t]
\centering
\includegraphics[width=3.6in]{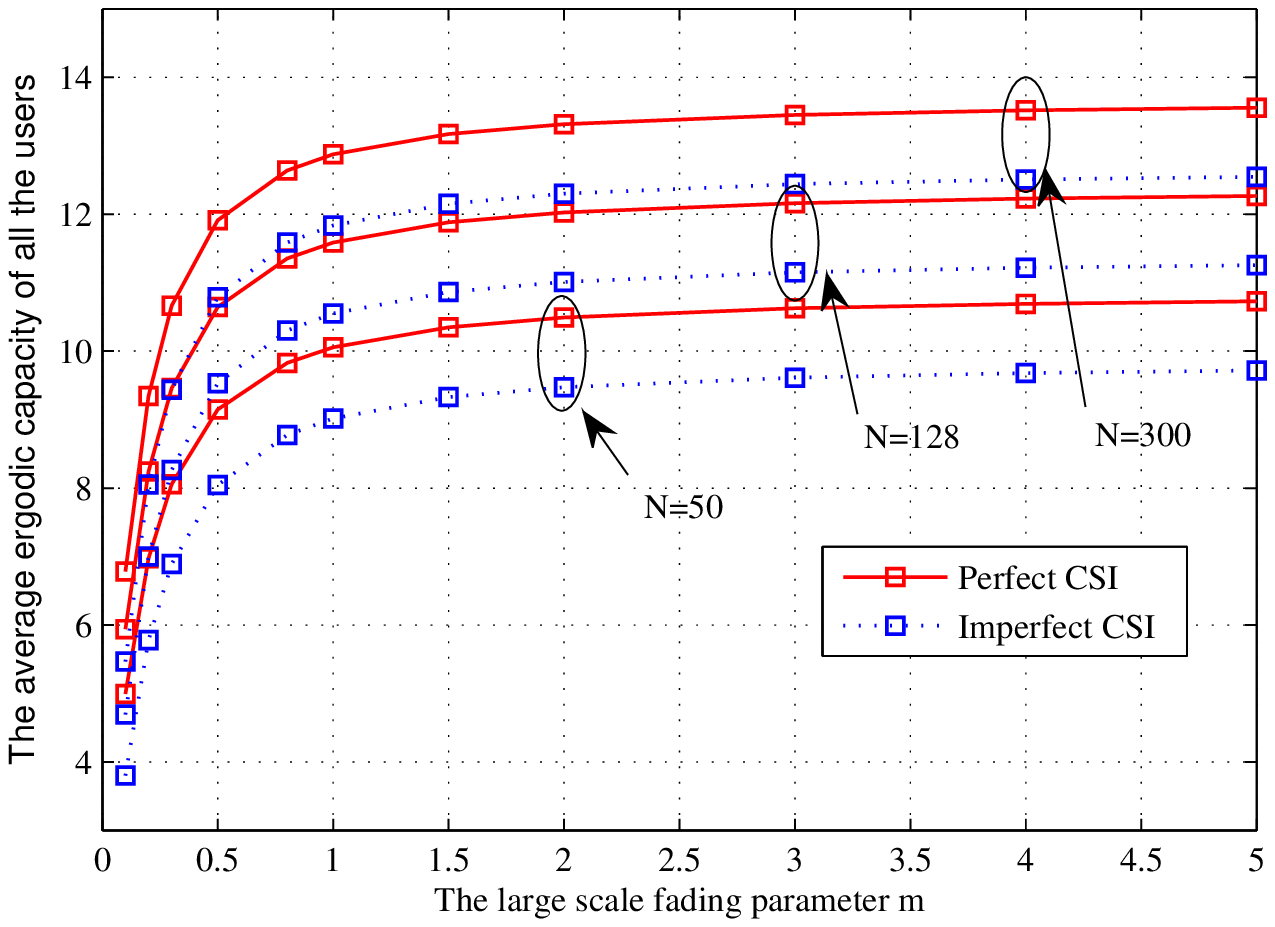}
\caption{The average ergodic capacity of all the users with
different large scale fading parameter $m$ in the case of $K=9$,
$N=50, 128, 300$, and $P=20dB$.}
\label{Capacity_different_naka_m_antennas.eps}
\end{figure}

\begin{figure}[!t]
\centering
\includegraphics[width=3.6in]{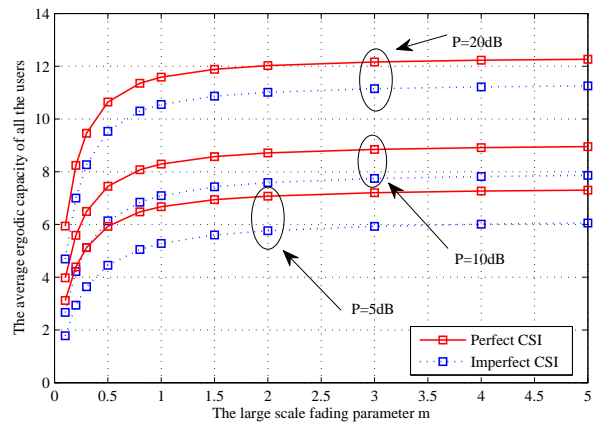}
\caption{The average ergodic capacity of all the users with
different large scale fading parameter $m$ in the case of $K=9$,
$N=128$, and $P=5dB, 10dB, 20dB$.}
\label{Capacity_different_naka_m_power.eps}
\end{figure}

\end{document}